\documentclass[letterpaper,12pt]{article}

\usepackage{dgjournal}

\usepackage{mathptmx}

\usepackage{graphics}

\usepackage[authoryear,comma,longnamesfirst,sectionbib]{natbib}

\usepackage[UKenglish]{babel}
\usepackage{authblk}

\usepackage{amssymb,amsfonts,amsmath}
\usepackage{amsthm,xypic,enumerate}
\usepackage{color,tikz,pgf}

\usepackage{booktabs}
\usepackage{mathtools}
\usepackage{soul}

\usepackage{bm}  
\usepackage{bbm}  

\usepackage{fancyhdr}
\usepackage{float}
\usepackage{hyperref}
\hypersetup{
    colorlinks=true, 
    linktoc=all,     
    linkcolor=blue,  
    citecolor=blue,
    urlcolor = blue,
    bookmarks=true,
    breaklinks=true,
    }
\usepackage{cleveref}

\makeatletter
\newcommand{\refcheckize}[1]{%
  \expandafter\let\csname @@\string#1\endcsname#1%
  \expandafter\DeclareRobustCommand\csname relax\string#1\endcsname[1]{%
    \csname @@\string#1\endcsname{##1}\wrtusdrf{##1}}%
  \expandafter\let\expandafter#1\csname relax\string#1\endcsname
}
\makeatother

\theoremstyle{plain}
\newtheorem{theorem}{Theorem}

\newtheorem{lemma}[theorem]{Lemma}
\newtheorem{corollary}[theorem]{Corollary}

\theoremstyle{definition}

\newtheorem{example}[theorem]{Example}

\theoremstyle{remark}

\renewcommand\SS{\mathbb S}

\DeclareMathOperator{\PP}{\mathbb{P}}

\DeclareMathOperator{\R}{\mathbb{R}}
\DeclareMathOperator{\E}{\mathbb{E}}

\DeclareMathOperator{\sD}{\mathcal{D}}
\DeclareMathOperator{\sE}{\mathcal{E}}

\DeclareMathOperator{\sT}{\mathcal{T}}

\DeclareMathOperator{\sV}{\mathcal{V}}
\DeclareMathOperator{\sW}{\mathcal{W}}
\DeclareMathOperator{\sX}{\mathcal{X}}

\newcommand{\rh}[1]{\left(#1\right)}
\newcommand{\vh}[1]{\left[ #1 \right]}
\newcommand{\set}[1]{\left\{ #1 \right\}}
\newcommand{\inpr}[1]{\left\langle#1\right\rangle}
\newcommand{\floor}[1]{\left\lfloor #1 \right\rfloor}

\DeclareMathOperator{\cov}{cov}

\DeclareMathOperator{\var}{var}
\DeclareMathOperator{\bias}{bias}

\DeclareMathOperator{\x}{\times}
\DeclareMathOperator{\im}{im}
\DeclareMathOperator{\diag}{diag}
\DeclareMathOperator{\op}{op}
\newcommand{\norm}[1]{\left\|#1\right\|}
\newcommand{\abs}[1]{\left|#1\right|}

\DeclareMathOperator{\en}{\quad\text{and}\quad}
\DeclareMathOperator{\sdot}{\,\cdot\,}

\DeclareMathOperator{\argmin}{arg min}

\begin{document}


\begin{center}
\Large{Estimation of the covariance structure from SNP allele frequencies}
\end{center}

\vspace{.5cm}
\begin{center}
\large{Jan van Waaij$^{1}$, Zilong Li$^{2}$, Carsten Wiuf$^{1,*}$}
\end{center}

\begin{center}
$^*$ Corresponding author (wiuf@math.ku.dk)
\end{center}

\vspace{3cm}
\noindent
Running title: Estimation of SNP Covariance Matrix

\vspace{1cm}
\noindent
$^1$ Department of Mathematical Science, University of Copenhagen, 2100 Copenhagen, Denmark

\noindent
$^2$ Department of Biology, University of Copenhagen, 2100 Copenhagen, Denmark

\newpage

\begin{abstract}
 We propose two new statistics, $\widehat V$ and $\widehat S$,  to disentangle the population history of related populations from SNP frequency data. If the populations are related by a tree, we show by theoretical means as well as by simulation that the new statistics are able to identify the root of a tree correctly, in contrast to standard statistics, such as the observed matrix of $F_2$-statistics (distances between pairs of populations). The statistic  $\widehat V$ is obtained by averaging over all SNPs (similar to  standard statistics). Its expectation is the true covariance matrix of the observed population SNP frequencies, offset by a matrix with identical entries. In contrast, the statistic $\widehat S$ is put in a Bayesian context and is obtained by averaging over  \emph{pairs} of SNPs, such that each SNP is only used once. It thus makes use of the joint distribution of pairs of SNPs.

In addition, we provide a number of novel mathematical results about old and new statistics,  and their mutual relationship.

 \end{abstract}

\section{Introduction}

A common situation in population genetics is ancestral disentanglement of related populations   \citep{PickrellPritchard2012,Pattersonea2012,Mailund2017,Lipson2020,Goldberg2021}. Imagine we observe genetic data in the form of allele frequencies from $n$ SNPs and $m$ related populations, and assume the population history is described by an unknown admixture graph. This graph  is  estimated from the data under the assumption of neutral evolution. 
The estimation typically takes place in two steps. First the covariance matrix of the SNP allele frequencies is estimated from the data, which in turn is used to determine the admixture graph.  This covariance matrix is at the core of much inference on population history. In this article we are interested in efficient estimation of the covariance matrix.

To put some notation, assume we observe  data vectors, $X^1,\ldots,X^n$ (one for each SNP), where $X^k=(X_{1}^k,\ldots,X_{m}^k)^t$ is an $m$-dimensional real-valued vector with common expectations $\E(X_{i}^k)=\mu_k$, $i=1,\ldots,m$, and $m\times m$ covariance matrix $\Sigma^k$, $k=1,\ldots,n$. Here, $X_{j}^k$ is the frequency of a particular allele (say, the reference allele) of the $k$th SNP  in population $j$. While it is standard to assume  an underlying admixture graph or tree \citep{Pattersonea2012,Lipson2020}, we will not  impose this here. However, we do assume the populations share a common ancestor (`root') at some point in the past, represented by the  mean value $\mu_k$.

The objective is to estimate
\begin{equation}\label{eq:1}
 \Sigma=\frac 1 n \sum_{k=1}^n \Sigma^k=\frac 1n \sum_{k=1}^n \E[(X^k-\mu_ke)(X^k-\mu_ke)^t],
\end{equation}
the average covariance matrix over all sites. The mean values $\mu_k$, $k=1,\ldots,n$, are  nuisance parameters of little interest.
In the absence of any knowledge about $\mu_k$, \cite{PickrellPritchard2012} suggests a surrogat  statistic $\widehat W$ for a related covariance  matrix $W$, obtained from $\Sigma$ by replacing $\mu_k$ in \cref{eq:1} with the average allele frequency. If the population history is a tree, one cannot infer the placement of the root from $\widehat W$.  Consequently, to rectify this, one might choose manually one population  as an outgroup and use this to place the root \citep{PickrellPritchard2012}. The same situation appears  for another surrogate statistic, the observed  distance matrix  $\widehat D$, that is an estimator of the  pairwise $F_2$-distance matrix $D$ (for formal definitions, see \cref{sec:covar}) \citep{Pattersonea2012}. 

In the present paper, we are concerned  with two things. The first is to make available some results on the  statistics $\widehat W$ and $\widehat D$,  and their mutual relationship. The second is to propose two new statistics, $\widehat V$ and $\widehat S$, that both can be used to recover the placement of the root, without using an outgroup. Whereas, $\widehat V$ is similar in spirit to $\widehat W$ and $\widehat D$ in the sense of averaging over all SNPs,  $\widehat S$ is based on pairwise comparison of SNPs, and is put in a Bayesian context. This statistic might open for new ways to explore the data. 

The results are stated  generally and do not rely on any specific distributional assumptions on the SNP allele frequencies. In particular, the $X_{i}^k$s do not need to be frequencies  at all, but could be arbitrary random variables with mean and variance. Hence, the proposed theory and methodology might have  wider applications in population genetics and genomics, as well as outside these fields.

\paragraph{Notation}

If $A$ is a matrix, then $A^t$  denotes the transposed matrix. Vectors are assumed to be column vectors. If $v$ is a vector, then $v^t$ is a row vector. 
Let $I$ be the $m\times m$ identity matrix, $E$ the symmetric $m\times m$ square matrix with all entries equal to one, and $e=(1,\ldots,1)^t$ the vector in 
$\R^m$ with all entries one.  Furthermore, let $e_i$ be the $i$th unit vector, $i=1,\ldots,m$.  So, $(e_i)_i=1$ and $(e_i)_j=0$ for $j\neq i$.  

For an $m\x m$ matrix $A$,  the Frobenius norm of $A$ is  $\|A\|_F=\sqrt{\sum_{a=1}^m\sum_{b=1}^m A_{ab}^2}$. 
For a linear operator $\sX\colon\mathbb{S}_m\to \mathbb{S}_m$,   the image is $\im(\sX)=\set{\sX(A)\colon A\in\mathbb{S}_m}$, and  the operator norm is
$$\| \sX\|_\text{op}=\sup_{\|A\|_F=1}\|\sX(A)\|_F.$$
If $\sX$ is an orthogonal projection then the operator norm is one. 

\section{Estimation of the covariance matrix}\label{sec:covar}

The theory to be developed holds for general random vectors, $X^1,\ldots,X^n$ with values in $\R^m$, $m\ge 2$. However, we put the theory in the context of population genetics as this is the application area we have in mind. Thus, we think of $X^k=(X_1^k,\ldots,X^k_m)$, $k=1,\ldots,n$,  as vectors of observed allele frequencies, either population or sample based.
 
Recall the covariance matrix in \cref{eq:1}, 
$$ \Sigma=\frac 1 n \sum_{k=1}^n \Sigma^k=\frac 1n \sum_{k=1}^n \E[(X^k-\mu_ke)(X^k-\mu_ke)^t].$$
 In the case  the means $\mu_k$, $k=1,\ldots,n$, are known, then a natural  unbiased estimator of $\Sigma=(\Sigma_{ij})_{i,j=1,\ldots,m}$, is
\begin{equation}\label{eq:definitionhatSigma}
\widehat \Sigma_{ij} = \frac1n \sum_{k=1}^n (X_{i}^k-\mu_k)(X_{j}^k-\mu_k),\quad i,j=1,\ldots,m.
\end{equation}
However, in  the absence of  such knowledge, we cannot  estimate $\Sigma$  from the data without further assumptions.  This has led to the proposal of alternative approaches, for example by substitution of $\mu_k$ with an estimated mean \citep{PickrellPritchard2012}. A natural unbiased estimator of $\mu_k$ is the moment estimator $\frac1m\sum_{i=1}^m X_i^k$. Plugging this into \cref{eq:definitionhatSigma}, yields the statistic $\widehat W$ given by 
\[\widehat W_{ij} = \frac1n\sum_{k=1}^n \rh{X_{i}^k - \widehat\mu_k}\rh{X_{j}^k- \widehat\mu_k},\quad\text{where}\quad \widehat\mu_k=\frac1m\sum_{i=1}^m X_{i}^k\]
is the empirical mean \citep{PickrellPritchard2012}. This is {\it not} an estimator of $\Sigma$ \emph{per se}, but it still contains information about the data generating process. In  \cite{PickrellPritchard2012},  $\widehat W$ is used as a surrogate for $\widehat\Sigma$.

Obviously, $\widehat W$ is a symmetric matrix and 
 \begin{align*}
\widehat W_{ij}  &= \frac1n\sum_{k=1}^n \rh{(X_i^k-\mu_k) - \frac1m \sum_{a=1}^m (X_{a}^k-\mu_k)}\rh{(X_j^k-\mu_k)- \frac1m \sum_{b=1}^m (X_{b}^k-\mu_k)}\\ 
 &= \widehat \Sigma_{ij} - \frac1m\sum_{a=1}^m \widehat \Sigma_{ia} -\frac1m\sum_{b=1}^m \widehat \Sigma_{jb}+\frac1{m^2} \sum_{a=1}^m \sum_{b=1}^m \widehat \Sigma_{ab}
\end{align*}
 \cite[equation 23]{PickrellPritchard2012}.

\cite{Pattersonea2012} suggest a different statistic to capture  the evolutionary distances between the populations. For populations $i$ and $j$, and  SNP $k$, the distance between the populations
 (at SNP $k$) is defined as the variance of $X_i^k-X_j^k$, which is known as an $F_2$ statistic. 
Let $D$ be the matrix with $(i,j)$ entry  $D_{ij}=\frac 1 n\sum_{k=1}^n\var(X_i^k-X_j^k)$.  An obvious estimator of $D$ is  defined by
\[
\widehat D_{ij} = \frac1n\sum_{k=1}^n (X_i^k-X_j^k)^2. 
\]
Also, $D$ and $\widehat D$ are symmetric matrices, and 
\[\widehat D_{ij} = \frac1n\sum_{k=1}^n (X_i^k-\mu_k+\mu_k-X_j^k)^2=\widehat \Sigma_{ii} + \widehat \Sigma_{jj} - 2\widehat \Sigma_{ij}.\]
Thus, both $\widehat D$ and $\widehat W$ are linear transformations of $\widehat \Sigma$. Furthermore, they are related to each other by an isomorphism (see \cref{lem:bijectionbetweenDandW}), and hence carry the same information. 
To formalise this, we need some further notation.

Let $\mathbb{S}_m$ be the vector space of symmetric $m\x m$-matrices with dimension $m(m+1)/2$, and define linear operators $\sD, \sW\colon\mathbb{S}_m\to \mathbb{S}_m$ by
\begin{align*}
\sW(A)_{ij} = & A_{ij} - \frac1m \sum_{a=1}^m A_{ia}- \frac1m \sum_{a=1}^m A_{ja}+ \frac1{m^2} \sum_{a=1}^m\sum_{b=1}^m A_{ab},\\
	\sD(A)_{ij} = & A_{ii}+A_{jj} -2 A_{ij}.
\end{align*}
 Obviously, $\widehat D= \sD(\widehat \Sigma)$ and $\widehat W = \sW(\widehat \Sigma)$. 

\begin{lemma}
The linear operator $\sW$ has the representation 
\begin{equation}\label{eq:WA}
\sW(A)= \left(I-\frac 1mE\right)A\left(I-\frac 1 mE\right).
\end{equation}
Consequently, if $A$ is positive definite, then $\sW(A)$ is positive semi-definite.  In particular, $\sW(\Sigma)$ is positive semi-definite.   
\end{lemma}

\begin{proof}
 The first part follows  by straightforward evaluation. For the second part, let $x\in\R^m$. We have 
\begin{align*}
    x^t\sW(A)x = x^t\left(I-\tfrac 1mE\right)A\left(I-\tfrac 1 mE\right)x= \left(\left(I-\tfrac 1mE\right)x\right)^tA\left(I-\tfrac 1 mE\right)x\ge 0, 
\end{align*}
as $A$ is positive semi-definite by assumption. The final part follows by noting that $\Sigma$ is positive definite, since it is a covariance matrix. 
\end{proof}

\begin{theorem}\label{lem:bijectionbetweenDandW}
	We have  
\[	\sD = \sD \circ \sW,\quad \sW = -\tfrac12\sW \circ \sD,\quad  \sD\circ \sD=-2\sD\en \sW\circ \sW=\sW.  
	\]
	The operator $\sW$ is an orthogonal projection (hence has operator norm one), 
 while $-\frac12\sD$  is a non-orthogonal  projection with operator norm 
$$ \|\!-\!\tfrac 1 2\sD\|_{\op}=\sqrt m.$$
	 The operators $\sD$ and $\sW$ have  the same $m$-dimensional kernel, given by 
\begin{align*}
K&=\set{ev^t+ve^t\colon v\in \R^m}, 
\end{align*}
	The restrictions
	\[\sW\colon\im(\sD)\to \im(\sW),\quad -\tfrac12\sD\colon\im(\sW)\to \im(\sD)\] 
are  each others inverse.  The images of $\sW$ and $\sD$ have dimension $\tfrac 1 2m(m-1)$. 
\end{theorem}
The proof of \cref{lem:bijectionbetweenDandW} is deferred to \cref{app:proofofbijectionbetweenDandW}.  

 Let $W:=\E(\widehat W)$ and $D:=\E(\widehat D)$.

\begin{theorem}
It holds that $W=\sW(\Sigma)$, and $D=\sD(\Sigma)$.
\end{theorem}

\begin{proof}
By linearity of the expectation and the definition of $\sW$ and $\sD$.
\end{proof}

The interpretation of the results are well understood in the case the   populations are related by a tree, see \cref{fig:exampletreewith2and3leaves}.  In this case, it is standard to associate  independent  random variables to the edges and the root of the tree, such that
\begin{equation}\label{eq:normal}
X^k_i = C^k_r + \sum_{e\in \Gamma_{ri}}C^k_e,\quad k=1,\ldots,n,\quad i=1,\ldots,m,
\end{equation} 
where the sum is over all edges $e$ on the unique path $\Gamma_{ri}$ from the root $r$ to population $i$,  $C^k_r$ is the random variable associated the root, and $C^k_e$ the random variable associated the edge $e$.  This model  naturally arises from the normal approximation model, 
$$X_t \sim \text{Normal}\left(x_0,\frac{t}{N_e}x_0(1-x_0)\right)$$
\citep{Donnelly2002},
where $x_0$ is the SNP frequency at time  $0$, $N_e$ is the effective population size, and $X_t$ is the SNP frequency $t$ generations later. The change in frequency   might be found by summing independent increments over different time epochs, leading to the model in \cref{eq:normal}.

Consider the case of $m=2$, and let $\Sigma$ be given as
\begin{equation}\label{eq:sigma}
\Sigma=\begin{pmatrix}
\sigma_1+\tau  & \tau \\ \tau & \sigma_2+\tau
\end{pmatrix},
\end{equation}
where $\sigma_1,\sigma_2,\tau\ge 0$, corresponding to the graph in \cref{fig:exampletreewith2and3leaves}(a). Typically, in an evolutionary context, $\tau$ might be taken to be zero (the variance of the root variable), as  data from the two populations will not contain any information about the evolution of the two populations prior to their most recent common ancestor. However, one might alternatively think of $\tau$ as the variance of the SNP means $\mu_k$, $k=1,\ldots,n$ (to be explored in \cref{sec:cross}).

 It follows from \cref{eq:sigma} and $D=\sD(\Sigma)$ that
$$\sD(\Sigma)=\begin{pmatrix}
0  &\sigma_1+\sigma_2\\ \sigma_1+\sigma_2 &0
\end{pmatrix},$$
from which only the sum $\sigma_1+\sigma_2$ might be recovered. Hence, neither the length of the ``root tip" ($\tau$) nor the placement of the root can be recovered. The elements of the kernel might be seen as operations on the tree, while preserving $\sD(\Sigma)$. This is perhaps best illustrated for $m=3$, in which case we take the tree of \cref{fig:exampletreewith2and3leaves}(b) as starting point, with
$$\Sigma=\begin{pmatrix}
\sigma_{11}  & 0 & 0 \\ 0 & \sigma_{12}+\sigma_2 &  \sigma_{12} \\ 0& \sigma_{12} & \sigma_{12}+ \sigma_3
\end{pmatrix}+\tau E.$$
Adding the kernel element
$$ e\begin{pmatrix}y&0&0\end{pmatrix}+\begin{pmatrix}y\\0\\0\end{pmatrix}e'+(x-y)E=%
xE+ \begin{pmatrix}y  &  0&0\\ 0 &-y&-y \\ 0& -y & -y\end{pmatrix}$$
to $\Sigma$ corresponds  to  extending the outgoing edge from the root with $x$ and sliding  the root by $y$ to the right on the edge, see \cref{fig:exampletreewith2and3leaves}(b) and (c). Choosing $y=\sigma_{12}$ yields a trifurcated star-shaped tree. Now, using other similar kernel elements, the star-shaped tree might be turned into many other trees while preserving the distance matrix
\begin{align*}\label{eq:D}
    &\sD(\Sigma)
= \begin{pmatrix}
        0 & \sigma_1+\sigma_2 & \sigma_1+\sigma_3 \\ \sigma_1+\sigma_2 & 0 & \sigma_2+\sigma_3 \\ 
        \sigma_1 +\sigma_3 & \sigma_2+\sigma_3 & 0
    \end{pmatrix},
\end{align*}
where $\sigma_{11}+\sigma_{12}=\sigma_1$.

The same applies for higher $m>3$ by iteratively applying kernel matrices to move ancestral nodes of the tree while preserving $\sD(\Sigma)$. 

 \begin{figure}
\centering{\includegraphics[width=\textwidth]{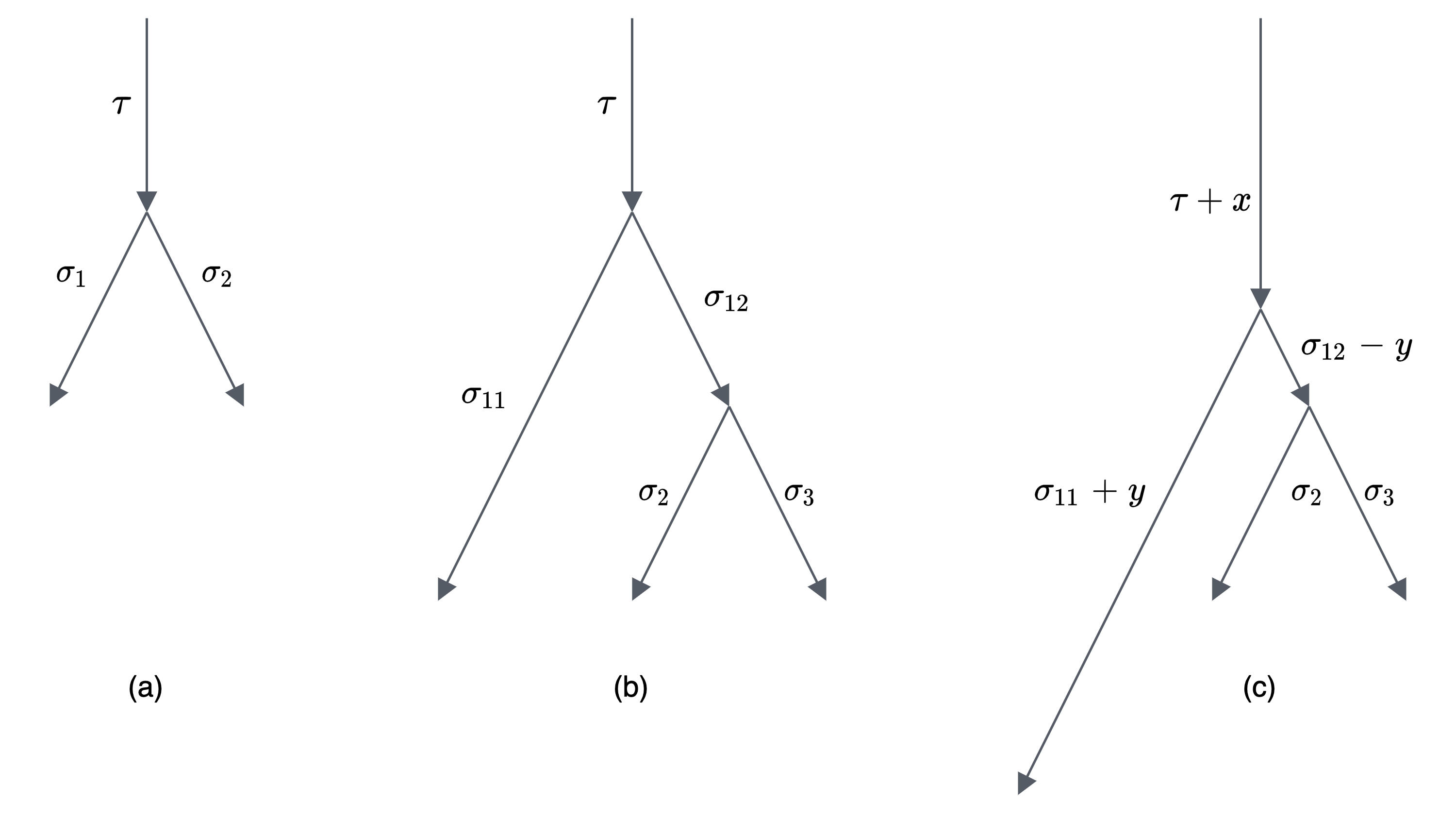}} \caption{A tree with two (a) and three (b) leaves. In (c), the length $\tau$ of the root tip is extended by $x$ and slided $y$ to the right.  We refer to the full \emph{edge labelled} tree as a \emph{rooted tree with a root tip}; the tree without the root tip but with the placement of the root as a \emph{rooted tree without the root tip}; and the tree without the root tip and the placement of the root as an \emph{unrooted tree}. In the latter case, the branches labeled $\sigma_{11}$ and $\sigma_{1}$ in (b) are replaced by a single branch of length $\sigma_{11}+\sigma_{12}$.} \label{fig:exampletreewith2and3leaves}
 \end{figure}

\section{A new statistic}

We suggest a third symmetric statistic, that carries more information about $\Sigma$ than the other two statistics. It is defined by 
\begin{align*}
\widehat V &= \frac1n\sum_{k=1}^n X^k(X^k)^t - \widehat\mu_k^2E.
\end{align*}
The second term is a correction term that makes the expectation of $\widehat V$ independent of the mean $\mu_k$. Also, it ensures the sum of all entries of $\widehat V$ is zero. In contrast to the other two statistics, $\widehat V$  is not linear in $\widehat\Sigma$. However, if we define $Y=\frac1n\sum_{k=1}^n X^k(X^k)^t $, then $\widehat W = \sW(Y)$, $\widehat D = \sD(Y)$ and $\widehat V = \sV(Y)$, where $V$ is defined in \cref{eq:definitionV}. 

By inspection, it holds that $\sW(\widehat V)= \widehat W$ and $\sD(\widehat V)=\widehat D$. Let $\sV\colon \mathbb{S}_m\to \mathbb{S}_m$ be defined as 
\begin{equation}\label{eq:definitionV}
\sV(A)=A - \frac1{m^2} EA E.
\end{equation}
 The following holds. 

\begin{theorem}
The operator $\sV$ is an orthogonal projection, $\sV=\sV\circ\sV$ with kernel $\set{\lambda E\colon \lambda\in\R }$. In particular, $\sV$ has operator norm one. 
 Furthermore, $V:=\E( \widehat V) =\sV(\Sigma)$.  
\end{theorem}

\begin{proof}
Let $V=\sV(A)$. Note that $e^tVe=0$. Hence, $\sV(V)=V$, which implies $\sV\circ\sV=\sV$. It follows from  \cref{eq:definitionV}, that $\sV(E)=0$, but also that $A$ differs from $\sV(A)$ by a constant times $E$. Hence $\ker(\sV)=\set{\lambda E:\lambda\in\R}$. 
Note that an $m\x m$-matrix $F$ is orthogonal to $E$ if and only if $\sum_{i=1}^m\sum_{j=1}^mF_{ij}=0$. We have
 \[\sum_{i=1}^m\sum_{j=1}^m \sV(A)_{ij}=\sum_{i=1}^m\sum_{j=1}^m A_{ij} - \sum_{i=1}^m\sum_{j=1}^m A_{ij}=0.\]
Hence, $\sV$ is an orthogonal projection. In particular, $\sW$ has operator norm one, for $m\ge 2$. 
Define  $\widehat V^k = X^k(X^k)^t - \hat\mu_k^2E$, such that $\widehat V = \frac1n\sum_{k=1}^n \widehat V^k$. 
Note that 
\begin{align*}
	\E (\hat\mu_k^2)=\frac1{m^2} \sum_{a=1}^m\sum_{b=1}^m\E(X_a^kX_b^k) \en 	\E(X^k(X^k)^t)= \Sigma^k + \mu_k^2E,
\end{align*}
 Then, $\E(\hat\mu_k^2)=\mu_k^2 + e^t\Sigma^k e/(m^2)$, and
\begin{align*}
	\E(\widehat V) = \frac1n\sum_{k=1}^n \Sigma^k - \frac1{m^2}(e^t\Sigma^ke) E=\Sigma - \frac1{m^2}E\Sigma E, 
\end{align*}
and the proof is complete.
\end{proof}

In the case of a tree,  the placement of the root is identifiable from $V=\sV(\Sigma)$, but not the length of the root tip. For $m=2$, we find
\begin{equation}\label{eq:m=2}
\sV(\Sigma)=\frac14\begin{pmatrix}
3\sigma_1-\sigma_2 &-\sigma_1-\sigma_2 \\ -\sigma_1-\sigma_2  & 3\sigma_2-\sigma_1
\end{pmatrix},
\end{equation}
from which $\sigma_1,\sigma_2$ might be recovered.   The general statement for arbitrary $m$ is here:

\begin{theorem}
    The  unrooted tree is identifiable from $D$ (or $W$). The rooted tree without the root tip is identifiable from $V$. The rooted  tree with the root tip is identifiable from $\Sigma$.   
\end{theorem}

\begin{proof}
The first statement is well known in literature  \cite[theorem 7.1.8, page 148]{semplesteel2003}. By \cref{lem:bijectionbetweenDandW}, it also holds true for $W$.

Let $\sT_1=(\sV_1,\sE_1)$  and $\sT_2=(\sV_2,\sE_2)$   be two rooted trees  with vertex  sets $\sV_1,\sV_2$, respectively, and edge sets $\sE_1$, $\sE_2$, respectively, and   common $m\times m$ covariance matrix $\Sigma$,.  Define the $(m+1)\x (m+1)$ matrix $A$ as follows 
\[
\begin{pmatrix}
\ddots & & &&\Sigma_{11}\\
&    \sD(\Sigma) & &&\vdots\\
&&\ddots &&\Sigma_{mm}\\
\Sigma_{11}&\ldots & \Sigma_{mm}& &0 \\
\end{pmatrix}.
\] 
The element $A_{ij}$ is the distance between the leaves $i$ and $j$ of $\sT_1$ (or $\sT_2$) for $i,j\le m$, and $A_{i,m+1}=A_{m+1,i}=\Sigma_{ii}$ is the distance from the root to the leaf $i$. If we consider the root  as another `leaf', then there is a isomorphism $\varphi\colon\sV_1\to \sV_2$, such that $x-y$ is an edge in $\sV_1$ if and only if $\varphi(x)-\varphi(y)$ is an edge in $\sT_2$ (indifferent of the direction). Moreover, the length of $x-y$ is equal to the length of $\varphi(x)-\varphi(y)$, and $\varphi$ maps the $i$th leaf of $\sT_1$ to the $i$th leaf of $\sT_2$, and the root of $\sT_1$ to the root of $\sT_2$. It follows that $\sT_1$ and $\sT_2$ are isomorphic as directed labelled trees.

Let $\sT_1$ and $\sT_2$ be two trees with the same matrix $V$. Let $\Sigma_1$ and $\Sigma_2$ be their corresponding covariance  matrices. Then, $\Sigma_1-\Sigma_2=\lambda E$ for some $\lambda$. Without loss of generality, $\lambda\ge 0$. If we make the root tip of $\sT_2$ $\lambda$ longer, resulting in a tree $\sT_2'$, with corresponding covariance matrix $\Sigma_2'$,  then $\Sigma_1-\Sigma_2'=0$. It follows that $\sT_1$ and $\sT_2'$ are isomorphic. Consequently, it follows that $\sT_1$ and $\sT_2$ are equal except for the length of the root tip.  
\end{proof}

 The following theorem relates $\sV$ with $\sD$ and $\sW$, analogous to \cref{lem:bijectionbetweenDandW}.  Simple examples show that $\sV\circ\sD\neq \sD$.

\begin{theorem}\label{thm:Vestimator}
It holds that $\sV\circ\sW=\sW\circ\sV=\sW$ and $\sD\circ\sV=\sD$. 
\end{theorem}
\begin{proof}
Let $A\in\mathbb{S}_m$ and let $V=\sV(A)$. Note that $A$ and $V$ differ only by a constant times $E$. As $E$ is in the kernel of $\sW$ and $\sD$, we have $\sW(V)=\sW(A)$ and $\sD(V)=\sD(A)$. This proves $\sD\circ\sV=\sD $ and $\sW\circ\sV = \sW$.

Note that $e^t\sW(A)e=e^t(I-E/m)A(I-E/m)e=0$. It follows that $\sV(\sW(A))=\sW(A)$. Hence, $\sV\circ\sW=\sW$. 
\end{proof}

We end by showing consistency of the statistic $\widehat V$, assuming (almost)  independence between sites, large $n$ and not too large $m$.

\begin{theorem}\label{thm:convergenceofhatV}
Assume $X^1,\ldots,X^n$ are random vectors in $\R^{m}$  with mean $\mu_k e$ and $m\x m$ covariance matrix $\Sigma^k$, and that there is  an integer $t\ge 1$, such that $X^k$ and $X^\ell$ are independent whenever $|k-\ell|\ge t$. Moreover, if there exists a constant $C>0$, such that the  forth moments of $X_i^k$, $k=1,\ldots,n$, $i=1,\ldots,m$,  are smaller than $C$, then 
$$\E(\|\widehat V-V\|_F)\le4\sqrt{\frac{m^2 t}n },$$
$$\E(\|\widehat D-D\|_F)\le 16\sqrt{\frac{m^2 t}n C},\quad \E(\|\widehat W-W\|_F)\le 16\sqrt{\frac{m^2 t}n C},$$
 for any $m,n,t$.
\end{theorem}

We defer the proof of \cref{thm:convergenceofhatV} to  \cref{subsec:proofthmconvergenceofhatV}.  
If $X^1,\ldots,X^n$, are frequencies, then the boundedness assumption is naturally met. The bound provides means to establish convergence in Frobenious norm as $n,m$ become large, and highlights the individual importance of $m, n, t$, respectively.

\section{Least Square Estimation}

In TreeMix \citep{PickrellPritchard2012}, the basic observation is $\widehat W$ from which parameters are estimated, for example, assuming  the populations are related by a tree. 
One might alternatively take $\widehat V$ to be the basic observation. We pose the question whether the parameter estimates obtained from $\widehat W$ and $\widehat V$, respectively,   are compatible?

A natural estimation procedure is Least Square (LS) estimation, which we will consider here. (We note that TreeMix in principle uses  weighted LS estimation, where  the weights are  empirically obtained.)
Let $L$ be a linear subspace of $\sV(\mathbb{S}_m)$ and $H\colon V\in L\subseteq \sV(\mathbb{S}_m)$   a linear hypothesis about $V$. 
We define the LS estimator of $V$ under  $H$ by
$$\widehat V_ L=\text{argmin}_{A\in L}\| \widehat V-A\|_F. $$
Similarly, one might estimate $W$  from $\widehat W$ under the corresponding linear hypothesis  $H'\colon W\in \sW(L)\subseteq \sW(\SS_m)$,
$$\widehat W_L=\text{argmin}_{A\in \sW(L)}\| \widehat W-A\|_F. $$
In either case, the LS estimator is the projection of the observation $\widehat V$ (respectively, $\widehat W$) onto the linear space $L$ (respectively, $\sW(L)$).

\begin{theorem}\label{thm:orthogonal}
    Let $L\subseteq\sV(\SS_m)$ be a linear subspace.     If $\sW(L)\subseteq L$, then $\widehat W_L=\sW(\widehat V_L)$. 
  Additionally, $\widehat W_L= \argmin _{A\in L}\|\widehat W- A\|_F$.
\end{theorem}

\begin{proof}As $\sW$ is an orthogonal projection, we have for $A\in L$, 
\begin{align*}\label{eq:LS_V}
    \|\widehat V - A\|_F^2  & =\|\sW(\widehat V - A)\|_F^2 + \|(I-\sW)(\widehat V-A)\|_F^2 \nonumber\\
 & =\|\widehat W - \sW(A)\|_F^2 + \|\widehat V-\widehat W - (I-\sW)(A)\|_F^2,
\end{align*}    
further using that $\sW(\widehat V)=\widehat W$.
As $\sW(L)\subseteq L$ by assumption, hence also $(I-\sW)(L)\subseteq L$, and $\sW(L)\oplus(I-\sW)(L)=L$ by orthogonality of $\sW$. Hence, the  minimum can be found as $\widehat V_L=\widehat V_1+\widehat V_2$, where
\[ \widehat V_1= \argmin _{A\in \sW(L)}\|\widehat W- A\|_F,\quad \widehat V_2= \argmin _{A\in (1-\sW)(L)}\|\widehat V-\widehat W-A\|_F.\]
This implies $\widehat V_1=\widehat W_L$ by definition and $\widehat W_L=\sW(\widehat V_L)$.

For the last statement, $\min_{A \in \sW(L)}\|\widehat W -A\|_F\ge \min_{A\in L}\|\hat W - A\|_F$, as $\sW(L)\subseteq L$. Using   orthogonality of $\sW$, we have for $A\in L$, 
\begin{align*}
    \|\widehat W- A\|_F^2= \|\widehat W- \sW(A)\|_F^2 + \|\sW(A) - A\|_F ^ 2.
\end{align*}
Further, for   $A\in L$,  $\sW(\sW(A))=\sW(A)$. Hence, $\|\widehat W-\sW(\sW(A))\|_F= \|\widehat W-\sW(A)\|_F$ and $\|\sW(\sW(A))-\sW(A)\|_F=0$. 
It follows that  $\|\widehat W -\sW(A)\|_F\le \|\widehat W -A\|_F$ for $A\in L$. Therefore,  $\min_{A \in \sW(L)}\|\widehat W -A\|_F\le  \min_{A\in L}\|\widehat W - A\|_F,$ and  consequently, $\min_{A \in \sW(L)}\|\widehat W -A\|_F= \min_{A\in L}\|\widehat W - A\|_F$.
\end{proof}

Note that $\sW(L)\subseteq L$ if and only if  $(I-\sW)(L)\subseteq L$. Hence, provided $\sW(L)\subseteq L$   holds,  it follows  from \cref{thm:orthogonal} by symmetry that the  LS estimator $\widehat U_L$  under the
linear hypothesis  $H''\colon U\in (I-\sW)(L)\subseteq \mathbb{S}_m$, 
$$\widehat U_L=\text{argmin}_{A\in (I-\sW)(L)}\| \widehat W-A\|_F. $$
fulfils $\widehat U_L=(1-\sW)(\widehat V_L)$. It leads to a reverse statement to that of 
\cref{thm:orthogonal}.

\begin{theorem}
Let $B$ be the support of the random variable $\widehat V=\widehat V(X_1,\ldots, X_n)$ and assume $\text{span}(B)=\sV(\mathbb{S}_m)$. Furthermore, let $L\subseteq\sV(\SS_m)$ be a linear subspace. If $\widehat W_L=\sW(\widehat V_L)$ and $\widehat U_L=(1-\sW)(\widehat V_L)$ hold for all $\widehat V\in B$,     then  $\sW(L)\subseteq L$.
\end{theorem}

\begin{proof}
We proceed by contradiction. By the remark above, we might assume  that $\sW(L)\not\subseteq L$ or $(I-\sW)(L)\not\subseteq L$, and show that it leads to a contradiction. Choose an arbitrary point $\widehat V\in B$ such that  $\widehat V=\widehat V_1+\widehat V_2\in\sW(L)\oplus(I-\sW(L))\setminus L$, 
where $\widehat V_1\in\sW(L)$, $\widehat V_2\in(1-\sW)(L)$. Such a point exists due to the  span condition. 

The LS estimate $\widehat V_L$ fulfils $\widehat V_L\not=\widehat V$, as $\widehat V\not\in L$, while the LS estimates   $\widehat W_L$  and $\widehat U_L$ clearly fulfil  $\widehat W_L=\widehat V_1$ and $\widehat U_L=\widehat V_2$, respectively, as $\widehat V_1\in \sW(L)$ and $\widehat V_2\in (I-\sW)(L)$ by assumption. 
 Since $\widehat V_L\not=\widehat V$, then either $\sW(\widehat V_L)\not=\widehat V_1$ or $(I-\sW)(\widehat V_L)\not=\widehat V_2$, contradicting the conditions of the theorem.  The proof is completed.
\end{proof}

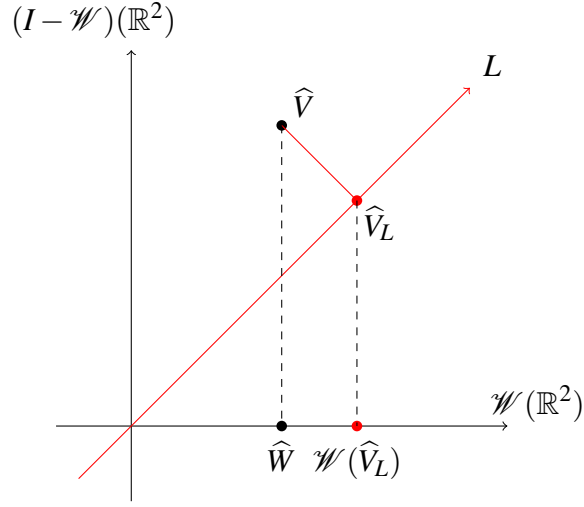
\begin{figure}[bt!] 
\begin{center}
\begin{tikzpicture}

\draw[->] (-1,0) to (5,0);
\draw[->] (0,-1) to (0,5);
\draw[->,red] (-0.7,-0.7) to (4.5,4.5);
\node  at (4.8,4.8)   {$L$};
\node  at (5.4,0.3)   {$\sW(\R^2)$};
\node  at (-0.5,5.4)   {$(I-\sW)(\R^2)$};
\fill[black]   (2,4) circle[radius=2pt];
\fill[red]   (3,3) circle[radius=2pt];
\node   at (2.3,4.3)   {$\widehat V$};
\node   at (3.3,2.7)   {$\widehat V_L$};
\draw[-,red] (3,3) to (2,4);
\fill[black] (2,0) circle[radius=2pt];
\fill[red] (3,0) circle[radius=2pt];
\node  at (2,-0.4)   {$\widehat W$};
\node  at (3,-0.45)   {$\sW(\widehat V_L)$};
\draw[-,dashed,black] (3,3) to (3,0);
\draw[-,dashed,black] (2,4) to (2,0);

\end{tikzpicture}
\end{center}

\caption{Imagine $\widehat V\in \R^2$, and that $\sW$ and $I-\sW$ are the projections onto the two coordinate axes, respectively. Furthermore, assume the hypothesis $H\colon V\in L$ corresponds to the red line. The corresponding hypothesis for $W$ is $H'\colon W\in \sW(L)=\R\times\{0\}$. The LS estimate of $ V$ under $H$ is $\widehat V_L$ (top red point), while the LS estimate of $W$ under $H'$, 
 is $\widehat W_L=\widehat W$ itself. However, this LS estimate is different from the projection of $\widehat V_L$ onto $\sW(L)$ (bottom red point).  Compared to the conditions of the theorem, $\sW(L)\not\subseteq L$. }
\label{fig:counter}
\end{figure}

\begin{example}
Two populations related by a tree as in \cref{fig:exampletreewith2and3leaves}(a) corresponds to the linear hypothesis, $H\colon V\in L$, given by
\begin{align*}
V\in L&=\left\{\frac14\begin{pmatrix}
3\sigma_1-\sigma_2 &-\sigma_1-\sigma_2 \\ -\sigma_1-\sigma_2  & 3\sigma_2-\sigma_1
\end{pmatrix}\,\Big|\, \sigma_1,\sigma_2\in\R\right\} \\
&=\left\{\frac14\begin{pmatrix} \sigma+2\delta &-\sigma \\ -\sigma   & \sigma -2\delta \end{pmatrix}\,\Big|\, \sigma,\delta\in\R\right\}, 
\end{align*}
where $\sigma=\sigma_1+\sigma_2$ and $\delta=\sigma_1-\sigma_2$, see \cref{eq:m=2}.
 The projection of $L$  by $\sW$ fulfils
\begin{align*}
\sW(L)&=\left\{\frac14\begin{pmatrix}
\sigma &-\sigma \\ -\sigma   & \sigma
\end{pmatrix}\,\Big|\, \sigma\in\R\right\}\subseteq L,
\end{align*}
Hence, \cref{thm:orthogonal}  applies.

However, assuming the branch lengths are related by $\sigma_1=2\sigma_2$,  then $\sigma=3\sigma_2$ and $\delta=\sigma_2$, and  $L$ reduces  to a one-dimensional linear subspace,
\begin{align*}
\widetilde L&=\left\{\frac14\begin{pmatrix}
5\sigma_2 &-3\sigma_2 \\ -3\sigma_2  &  \sigma_2
\end{pmatrix}\,\Big|\, \sigma_2\in\R\right\},
\end{align*}
while $\sW(\widetilde L)=\sW(L)$. Clearly, 
$$\sW(\widetilde L)\cap \widetilde L=\left\{\begin{pmatrix}
0 &0 \\ 0  &  0\end{pmatrix}\right\}\not=\sW(\widetilde L),$$
and the conclusion of \cref{thm:orthogonal} does not hold.
This case might be seen as an instance of \cref{fig:counter}.
\end{example}

\section{Combining information across SNPs}
\label{sec:cross}

By combining information across SNPs, one might derive more informative about the data  generating process and also derive other useful statistics. In this case, it is necessary to require some regularity across sites for reasons of comparison. We propose one such statistic, which is closely related to $\Sigma$ in the previous section, by 
\begin{align*}
	\widehat S = \frac1{2\floor{n/2}}\sum_{k=1}^{\floor{n/2}} (X^{2k}-X^{2k-1})(X^{2k}-X^{2k-1})^t.
\end{align*}
assuming the number of SNPs is even (if it is odd, one might discard one SNP).   Assuming the true allele frequencies are draws from a common distribution then the average allele frequency cancels out in the difference $X^{2k}-X^{2k-1}$.
Thus, we are left with an expression for the variance alone, see below.

As $\widehat S$ makes use of information from pairs of variables, it is natural to impose some  regularity conditions on the parameters $(\mu_k,\Sigma^k)$, $k=1,\ldots,n$, of the model. Perhaps the simplest approach is to embed the model into a Bayesian framework (as is often used for simulation purposes \citep{Posada2016}).  Specifically, we assume $(\mu_k,\Sigma^k)$, $i=1,\ldots,n$, are draws (at this point not necessarily independent) from a common distribution $F$, and the random vector $X^k$ subsequently is a draw from a distribution $G$, characterised by $(\mu_k,\Sigma^k)$,
\begin{equation}\label{eq:modelrandomSigmaandmu}
\begin{split}
	(\mu_k,\Sigma^k)&\sim F \\
	X^k\mid \mu_k,\Sigma^k &  \sim G(\sdot \mid \mu_k,\Sigma^k).
	\end{split}     
\end{equation}
Here, we assume $F$ is a distribution concentrated on $\R\x \PP_m$, where $\PP_m\subseteq\mathbb{S}_m$ is the space of real symmetric positive definite matrices  with mean $(\mu_0,\Sigma_0)$, and the marginal distribution of $\mu_k$ has variance $\tau$.

Then, $X^k$ has mean,
\begin{align*}
	\E (X^k) = \E(\E(X^k\mid \mu_k,\Sigma^k))=\E(\mu_ke)=\mu_0  e,
\end{align*}
and covariance 
\begin{align*}
	\cov(X^k)  & =\E[(X^k-\mu_0 e)(X^k-\mu_0 e)^t]\\
	& = \E\big(\E[(X^k-\mu_ke+\mu_ke-\mu_0 e)(X^k-\mu_ke+\mu_ke-\mu_0 e)^t\mid \mu_k,\Sigma^k]\big)\\
	& = \E(\Sigma^k) +\E((\mu_k-\mu_0)^2E)\\
	& = \Sigma_0 + \tau E,
\end{align*}
where  $\tau=\E((\mu_k-\mu_0)^2)$. Set $\Sigma_1=\Sigma_0 + \tau E$.  
Since $\Sigma^k$ is assumed to be positive definite, then so is $\Sigma_0$, and hence also $\Sigma_1$. The latter  follows directly from $x^t(\Sigma_0 + \tau E)x=x^t\Sigma_0 x+ \tau (\sum_{i=1}^mx_i)^2\ge 0$ (with equality if and only of $x=0$).

 Assuming  $X^{2k-1}$ and $ X^{2k}$ are independent, then  
\begin{align*}
&\E[(X^{2k}-X^{2k-1})(X^{2k}-X^{2k-1})^t]\\ & \quad= \E[(X^{2k}-\mu_0 e+\mu_0 e-X^{2k-1})(X^{2k}-\mu_0 e+\mu_0 e-X^{2k-1})^t]\\
	&\quad= 2 \Sigma_1,
\end{align*}
hence $E(\widehat S)=\Sigma_1$.

To connect to the model of Section \ref{sec:covar}, we might think of $\Sigma$ as  $\Sigma_1=\Sigma_0+\tau E$, and $\tau$ as the variance of the means across sites. 

\medskip
In the context of population genetics, the assumption that $X^{2k-1}$ and $X^{2k}$ are independent, is quite mild. We only ask for a pairing of the variables, $X^1,\ldots,X^n$, such that the two variables of each pair are independent, not that  pairs of variables themselves are  independent. One could, for example take one member of the pair from one chromosome and the other from another chromosome, assuming there are sufficient number of SNPs for such pairing. A precise condition is given here.

\begin{lemma}
Assume that each SNP with a corresponding random variable is associated to one of $C$ chromosomes, such that random variables  associated to SNPs on different chromosomes are  independent of each other. Let $n_i$ be the number of SNPs associated to chromosome $i$, $i=1,\ldots,C$. Furthermore, assume the chromosomes are ordered such that $n_1\ge n_2\ge\ldots\ge n_C$. If $n_1+\ldots+n_C$ is an even number and $\sum_{i=2}^Cn_i\ge n_1$, then the SNPs can be ordered in pairs, such that the corresponding random variables  of each pair are independent.
\end{lemma}

A proof can be found in \cite{Hakemi}. A multi-graph (a graph potentially with multiple edges between two nodes) is constructed with $C$ nodes, representing chromosomes. Each edge between two nodes represents a pair of variables. Then there is a simple automated method for ordering the pairs: the $n_C$ variables on chromosome $C$ are linked to $n_C$ variables on chromosome $1$. Then, there are $n_2,\ldots, n_{C-1}$ and $n_1'=n_1-n_C$ variables left on $C-1$ chromosomes. These are reordered from large to small  and the pairing reiterated  \cite{Hakemi}.

The proof of the next statement can be found in \cref{proof:convergenceinL2ofwidetildeSigma}. 

\begin{theorem}\label{thm:convergenceinL2ofwidetildeSigma}
    Assume $X^1,\ldots,X^n$ are random vectors in $\R^{m}$  defined by \cref{eq:modelrandomSigmaandmu}, and that there exists  an integer $t\ge 1$, such  that the pairs  $(X^{2k-1},X^{2k})$ and $(X^{2\ell-1},X^{2\ell})$ are independent whenever $|k-\ell|\ge t$, and that $X^{2k-1}$ and $X^{2k}$ are  independent for $k=1,\ldots,\floor{n/2}$. Moreover, if there exists a constant $C>0$, such that the  forth moments of $X_i^k$, $k=1,\ldots,n$, $i=1,\ldots,m$,  are smaller than $C$, then 
\begin{align*}
\E[\|\widehat S-\Sigma\|_2] &\le 4\sqrt{\frac{m^2t}{\floor{n/2}} C},\end{align*}
for all $n,m$. 
\end{theorem}

This estimator has as the additional benefit that it  accurately estimates
the variance of $X^k$, while $\widehat V$ only estimates it up to a constant.

\subsection{Sampling bias}

In the previous section, we did not make any specific assumptions about the random vectors $X^1,\ldots,X^n$, though it would  be natural to think of them as  population allele frequencies. 
However, typically, we do not have access to population allele frequencies, but only sample allele frequencies.

 To make this specific, let  $X^1,\ldots,X^n$ denote population allele frequencies and 
 $X^{s,k}=(X_{1}^{s,k},\ldots,  X_{m}^{s,k})$, $k=1,\ldots,n$, be the corresponding sample allele frequencies. We will assume the sample  allele counts  are binomial, that is, $ X_{i}^{s,k}=Z_{i}^k/(2N_{ik})$, where $Z_{i}^k
\sim \text{Bi}(2N_{ik},X_{i}^k)$, and $N_{ik}$ denotes the  sample size  at site $k$ in population $i$. By  allowing $N_{ik}$ to vary over $k$, we allow for missing data   across loci. 

Define 
\[Y= \frac1n \sum_{k=1}^n X^k(X^k)^t,\quad   Y^s = \frac1n \sum_{k=1}^n  X^{s,k} (X^{s,k})^t.\]
Then, the  three statistics  $\widehat W$, $\widehat D$ and $\widehat V$ are linear maps of $Y$, namely,  $\widehat  W=\sW(Y)$, $\widehat  D=\sD(Y)$ and $\widehat  V=\sV(Y)$ (the proof is left to the reader).
 Conditioned on $X^k$, the variable $X^{s,k}-X^k$ has zero mean, such that
 \begin{align*}
     \E[ X^{s,k} (X^{s,k})^t] = \E [( X^{s,k}-X^k)( X^{s,k}-X^k)^t] + \E[X^k (X^k)^t],
 \end{align*}
by adding and subtracting $X^k$.

Also conditioned on $X^k$, the sample variables $ X^{s,k} _a$ and $ X^{s,k}_b$ are independent for $a\neq b$. Hence, $\E[( X^{s,k}_a - X^k_a)( X^{s,k} _b - X^k_b)]=0$ for $a\neq b$. Furthermore, 
$$\E[(\hat X^k_a - X^k_a)^2]= \frac{ X^k_a(1-X^k_a)}{2N_{i,k}}.$$
 Thus, the bias correction of  $ Y_n^s$ is the diagonal matrix  
\[
\text{bias}( Y^s) = \frac1n\sum_{k=1}^n\diag\rh{\frac{X_1^{s,k}(1-X_1^{s,k})}{8N_{1k}^2(N_{1k}-1)},\ldots,\frac{X_m^{s,k}(1-X_m^{s,k})}{8N_{mk}^2(N_{mk}-1)}}
\] 
 \cite[text S1, supplementary material]{PickrellPritchard2012}. 

By the linearity of the mean (and hence the bias) the bias of $\widehat D, \widehat W$ and $\widehat V$ are  $\sD(\text{bias}(Y^s)),$ $\sW(\text{bias}( Y^s))$, and $\sV(\text{bias}( Y^s)),$ respectively. 

\medskip
Similarly, the bias correction for $\widehat S$ is
\begin{align*}
\bias(\hat S ) = \frac1{2\floor{n/2}}\sum_{k=1}^{\floor{n/2}}\diag\Bigg(\frac{X_1^{s,2k-1}(1-X_1^{s,2k-1})}{8N_{1,2k-1}^2(N_{1,2k-1}-1)}+\frac{X_1^{s,2k}(1-X_1^{s,2k})}{8N_{1,2k}^2(N_{1,2k}-1)},\ldots,\quad \quad &\\
\frac{X_m^{s,2k-1}(1-X_m^{s,2k-1})}{8N_{m,2k-1}^2(N_{m,2k-1}-1)}+\frac{X_m^{s,2k-1}(1-X_m^{s,2k-1})}{8N_{m,2k-1}^2(N_{m,2k-1}-1)}\Bigg).&
\end{align*}

 \section{Simulation results}

Here we present simulation results and analyses of real data that show one may 
identify the position of the root in a genealogical tree from both $\widehat V$ and $\widehat S$ directly. This is in contrast to TreeMix that relies on an outgroup  to place the root onto the tree.

For each of the scenarios below, we compute $\widehat W$, $\widehat V$, and $\widehat S$, as well as run TreeMix  by  specifying an outgroup. To estimate   the placement of the root  from  $\widehat V$ and $\widehat S$, respectively, we simply search for the partition of the $m$ populations into two groups that minimizes the average covariance between populations in different groups. The rationale for this  is that the   covariance $\text{Cov}(X^k_i,X^k_j)$, the $(i,j)$th entry of $\Sigma^k$, is smallest among the covariances when population $i$ and $j$ descend from opposite branches emanating from the root. The same holds for the $(i,j)$th entry of $V$ and $\Sigma_1$ (the expectation of $\widehat S$).

\subsection{Two simulation scenarios}

We adopt a test scenario used in  \cite{PickrellPritchard2012} and originally proposed in \cite{DeGiorgio2009} to study human evolution. We consider 20 populations related by a tree as shown in \cref{fig:tree}. At each split in the tree, the ‘outbranching’ ancestral population goes 
through a bottleneck, but population sizes are otherwise constant. We simulated two scenarios using 
the same commands as in  \citet[p4 of the supplementary information]{PickrellPritchard2012}, a short branch and a long branch scenario. Specifically, we assume
\begin{itemize}
\item 200 Mb long genome distributed into 400 independent regions, each 500 Kb long,
\item 20 individuals sampled from each of the 20 populations,
\item  Time and parameters are scaled by the effective population size, see \cite{Hudson1983,Hudson2002} for details, using an effective population size of $N_e = 10000$, and a per base per generation mutation/recombination rate of $10^{-8}$. This yields a population scaled mutation rate of $\theta=200$, and population scaled recombination rate of $\rho=200$ for each region,
\item Splits happen at equidistant times, the $i$th population splits out from the $(i-1)$th population at time $T(21-i)$, $i=2,\ldots,20$, in the past. In the short branch  scenario $T=0.00275$; in the  long branch  scenario $T=0.1375$ (50 times longer than in the short branch scenario),
\item Immediately after the $i$ population has split from the $(i-1)$th population, its population size is reduced to $2.5\%$ of its original size. The bottleneck lasts for $B$ time units before regaining its original size. In the short branch  scenario $B=0.00005$; in the  long branch  scenario $B=0.0025$ (50 times longer than in the short branch scenario).
\end{itemize}

\begin{figure}[bt!] 
\begin{center}
\begin{tikzpicture}
\draw[-,red,line width=1mm] (0,5) to (1,5);
\draw[-,red,line width=1mm] (1,4) to (2,4);
\draw[-,red,line width=1mm] (4,2) to (5,2);
\draw[-,red,line width=1mm] (5,1) to (6,1);
\draw[-,red,line width=1mm] (1,4.5) to (1,5);
\draw[-,red,line width=1mm] (2,3.5) to (2,4);
\draw[-,red,line width=1mm] (5,1.5) to (5,2);
\draw[-,red,line width=1mm] (6,0.5) to (6,1);
\draw[-,gray,line width=4mm] (0,0) to (0,5.3);
\draw[-,gray,line width=4mm] (1,0) to (1,4.5);
\draw[-,gray,line width=4mm] (2,0) to (2,3.5);
\draw[-,gray,line width=4mm] (5,0) to (5,1.5);
\draw[-,gray,line width=4mm] (6,0) to (6,0.5);

\node at (3,2) {\ldots};
\node at (3.5,-0.5) {\ldots};
\node  at (0,-0.5)   {$1$};
\node  at (1,-0.5)   {$2$};
\node  at (2,-0.5)   {$3$};
\node  at (5,-0.5)   {$19$};
\node  at (6,-0.5)   {$20$};
\end{tikzpicture}
\end{center}

\caption{Schematic drawing of the simulation set-up. Sequential splitting at equidistant times. After each splits the one of the populations undergoes a severe bottleneck. The placement of the root can be identified from the two groups of populations descending from the two branching emanating from the root; here population 1 and populations 2-20.}
\label{fig:tree}
\end{figure}
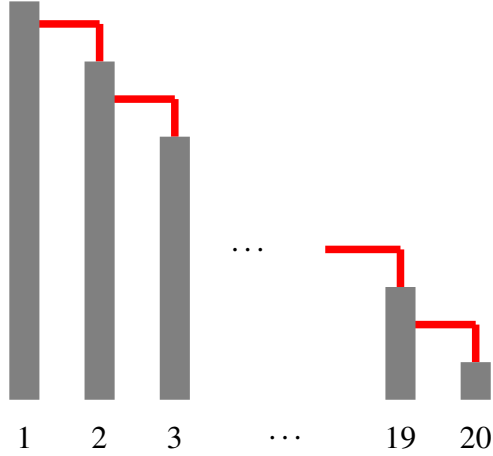

The simulation results in 1,225,747 SNPs  in the short branch scenario, and 6,530,862 SNPs in the long branch scenario. Since we simulate  a large number of SNPs, we do not bias correct. 

We compute the covariance $\Sigma$  assuming the normal approximation and a fixed root frequency $x^k_0$ for SNP $k=1,\ldots,n$, see \cref{eq:normal}. Then, the entries become
\begin{align}\label{eq:theor}
\Sigma_{ij}^k&=(i-1)\left(T-B+\frac{B}{0.025}\right)x^k_0(1-x^k_0),\quad\text{for}\quad 1\le i<j\le m  \\
\Sigma_{ii}^k&=\left[(i-1)\left(T-B+\frac{B}{0.025}\right)+(21-(i+1))T\right]\!x^k_0(1-x^k_0),\quad\text{for}\quad i=1,\ldots,m. \nonumber 
\end{align}
The variance $\Sigma_{ii}^k$ increases with increasing $i$. The covariance $\Sigma_{ij}^k$ is independent of $j>i$, and increases with increasing $i$. The difference between $\Sigma$ and $V$ is a constant matrix, hence the same conclusions hold for $V$.

Using population $1$ as an outgroup, Treemix  constructs the tree topology  exactly as modeled. However,  if there is not an outgroup specified or a wrong outgroup is used, then Treemix cannot return the correct tree topology. With our statistics $\widehat V$  and $\widehat S$, we  correctly   identify the split into one group consisting of 
population 1 and another group consisting of the remaining populations, both in the short as well as the long branch scenario, see \cref{fig:ori.cov1} and \cref{fig:ori.cov2}.

 \begin{figure} 
\centering{\includegraphics[width=\textwidth]{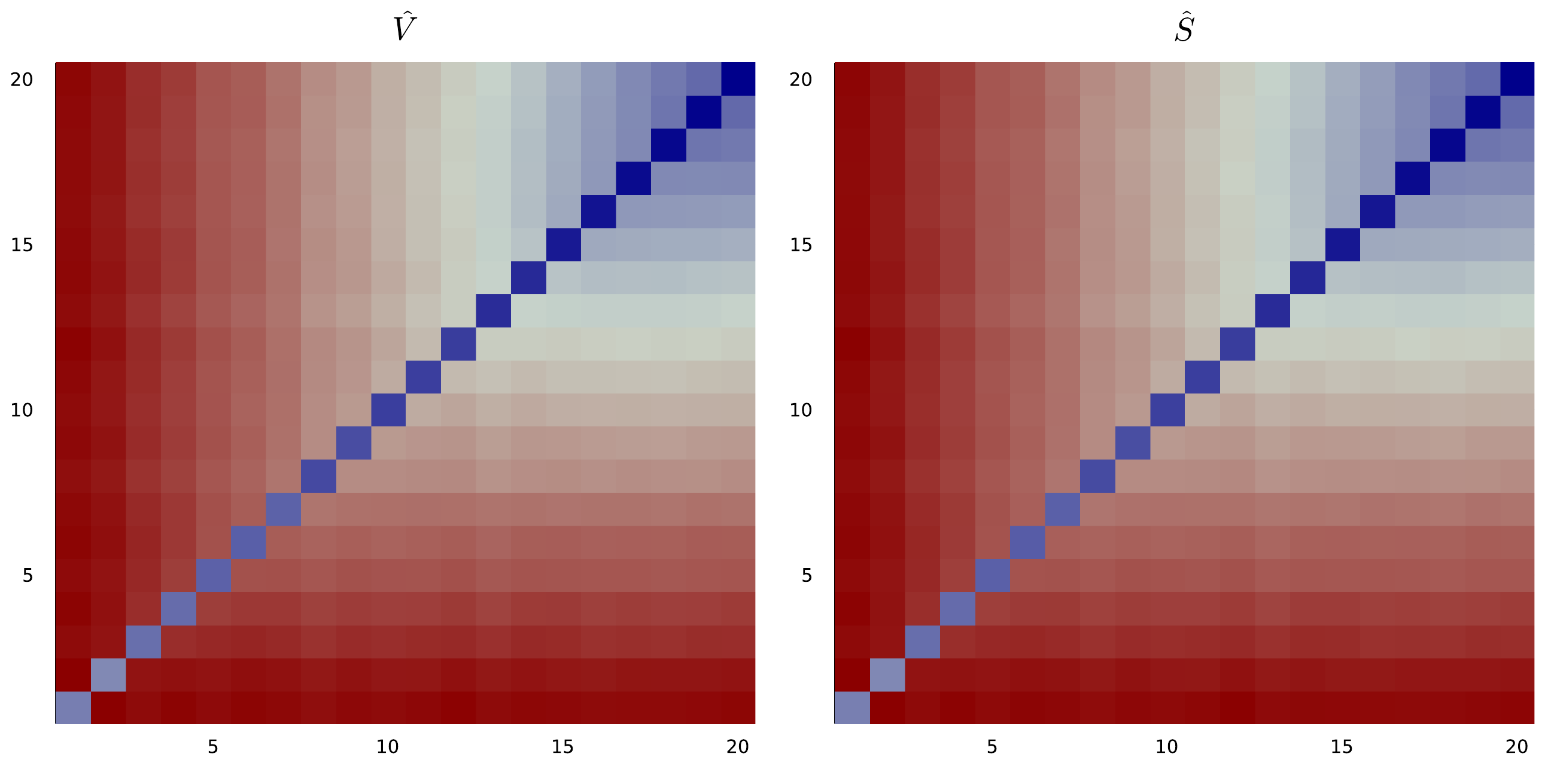}} \caption{Short branch scenario.  Color-coding runs from dark red (small values) to dark blue (large values). The diagonal elements $\widehat V_{ii}$ ($\widehat S_{ii}$) increases from population 1 to 20 due to an increasing number of bottlenecks (zero for population 1; 19 for population 20). Also, the off-diagonal elements $\widehat V_{ij}$ ($\widehat S_{ij}$) are roughly constant for  $j>i$. Both observations are in accordance with theoretical expectations, \cref{eq:theor}. } \label{fig:ori.cov1}
 \end{figure}

 \begin{figure} 
\centering{\includegraphics[width=\textwidth]{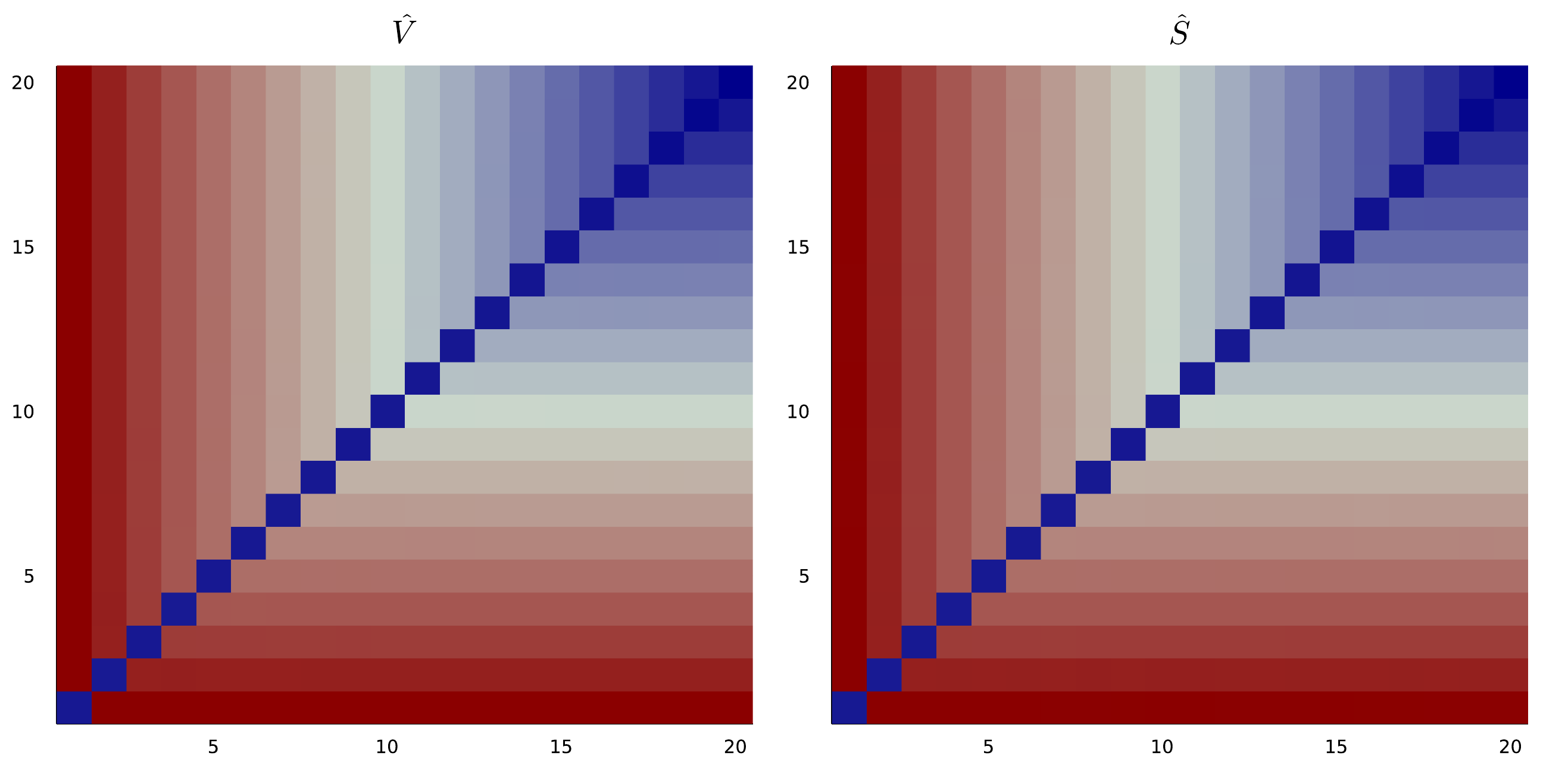}} \caption{Long branch scenario.   Color-coding runs from dark red (small values) to dark blue (large values). } \label{fig:ori.cov2}
 \end{figure}

\subsection{Data from the 1000 Genomes Project}

We selected data from six populations from the 1000 Genomes Project (see https:// www.internationalgenome.org/data-portal/data-collection/30x-grch38) that   are supposedly not admixed: YRI (Yoruba in Ibadan, Nigeria; 108 individuals), LWK (Luhya  in Webuye, Kenya; 99 individuals), CEU (Northern and Western European; 99 individuals), FIN (Finnish; 99 individuals), CHB (Han Chinese; 103 individuals), CDX (Dai Chinese; 93 individuals).
The number of SNPs is 4,391,887; all SNPs with MAF $>5\%$. Since the data set contains a large number of SNPs, we do not bias correct. 

Using YRI as an outgroup, TreeMix produces the tree in \cref{fig:1000Genomes1}. In contrast, using either $\widehat V$ or $\widehat S$, we identify the root to separate the clades (YRI, LWK) and (CEU, FIN, CHB, CDX), see \cref{fig:1000Genomes2}. Placing the root between the two clades would produce a more balanced, molecular clock-like tree.

 \begin{figure}[htp]
\centering{\includegraphics[width=\textwidth]{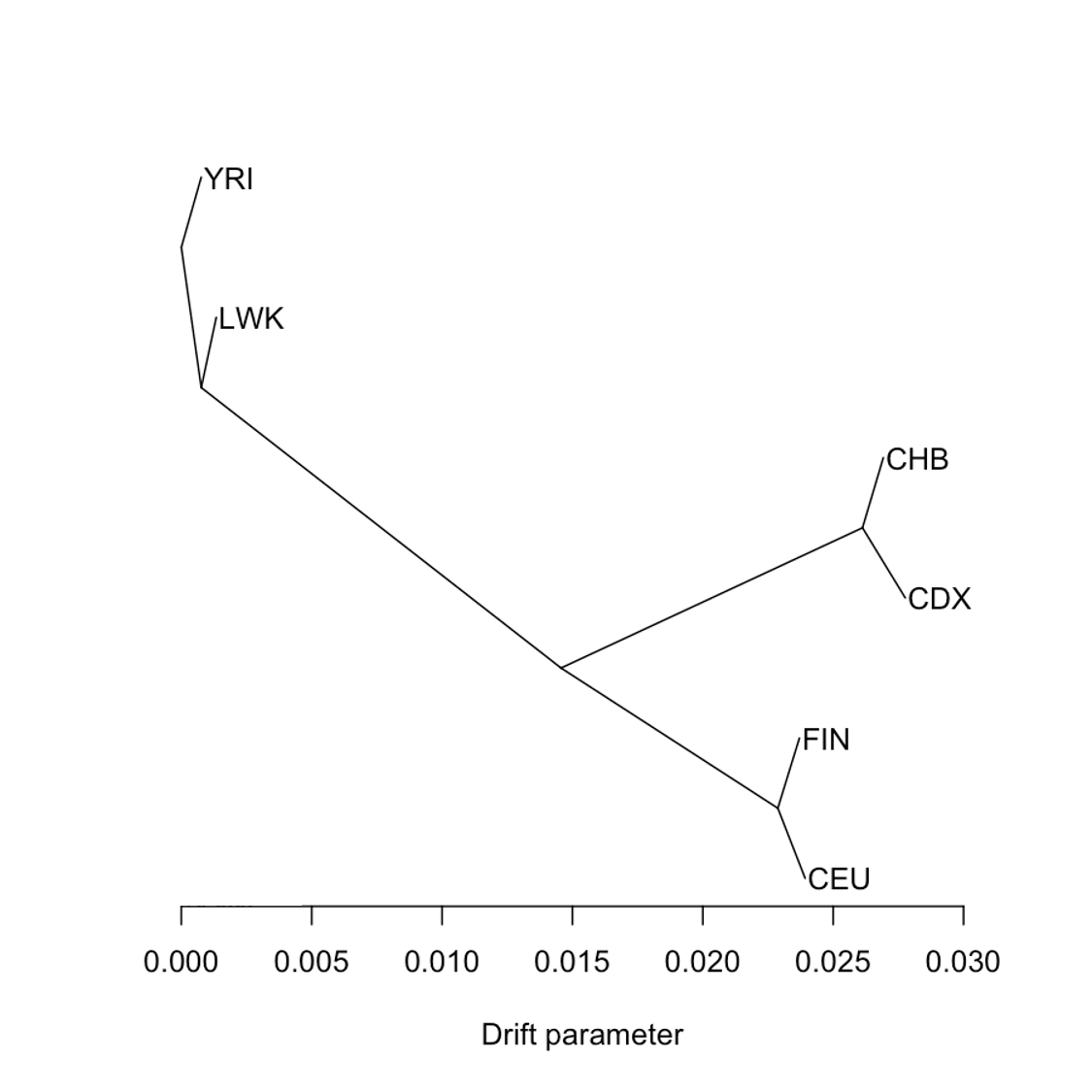}} \caption{TreeMix tree of the six 1000 Genomes Project populations, forcing YRI to be the outgroup.} \label{fig:1000Genomes1}
 \end{figure}

 \begin{figure} 
\centering{\includegraphics[width=\textwidth]{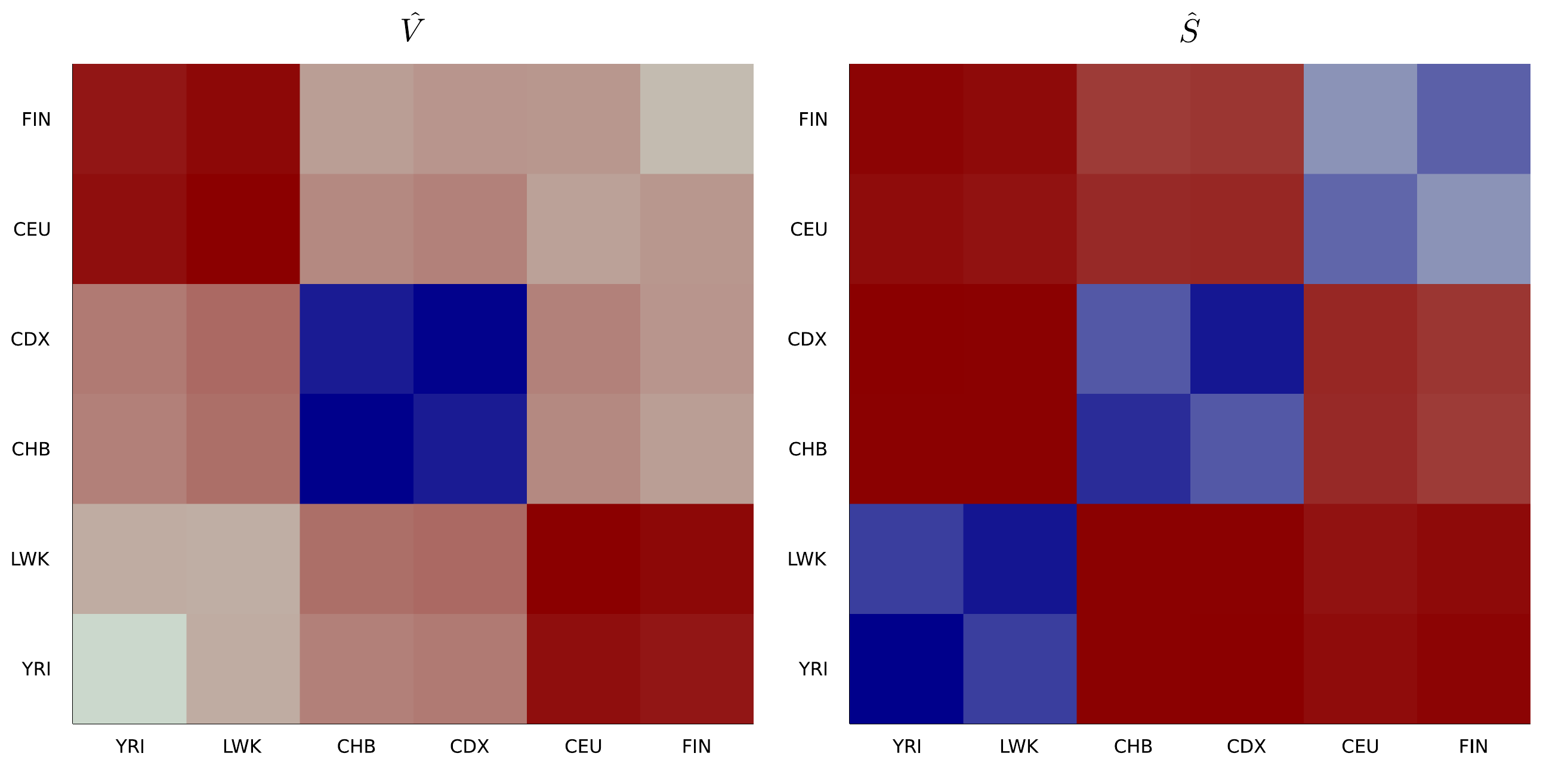}} \caption{Color-coding runs from dark red (small values) to dark blue (large values). Population 1: YRI, 2: CHB, 3: CDX, 4: LWK, 5: CEU, 6: FIN. For both statistics, the entries for pairs of populations in the two different clades, (YRI, LWK) and (CEU, FIN, CHB, CDX), are smaller than any other entry. } \label{fig:1000Genomes2}
 \end{figure}

\section{Proofs}

\subsection{Proof of \cref{lem:bijectionbetweenDandW}}\label{app:proofofbijectionbetweenDandW}
	Let $A$ be a symmetric $m\x m$ matrix. Let $W=\sW(A)$. So \begin{align*}
		\sD(W)_{ij}  & =W_{ii} + W_{jj} - 2W_{ij}\\
		 & = A_{ii} -\frac2m \sum_{k=1}^m A_{ik}  + \frac1{m^2}\sum_{k=1}^m\sum_{\ell=1}^m A_{k\ell}\\
		&\quad + A_{jj} -\frac2m \sum_{k=1}^m A_{jk} + \frac1{m^2}\sum_{k=1}^m\sum_{\ell=1}^m A_{k\ell}\\ 
		& \quad- 2\rh{A_{ij} -\frac1m \sum_{k=1}^m A_{ik} - \frac1m \sum_{k=1}^m A_{jk} + \frac1{m^2}\sum_{k=1}^m\sum_{\ell=1}^m A_{k\ell}}\\
		 &= A_{ii} + A_{jj} - 2A_{ij}\\
		 &= \sD(A)_{ij}. 
	\end{align*}
	It follows that $\sD= \sD\circ \sW$. 
	
	Let $D=\sD(A)$. Then \begin{align*}
		-\frac12\sW(D)_{ij}  & =-\frac12D_{ij} +\frac1{2m} \sum_{k=1}^m D_{ik} + \frac1{2m} \sum_{k=1}^m D_{jk} - \frac1{2m^2}\sum_{k=1}^m\sum_{\ell=1}^m D_{k\ell}\\
		 &=-\frac12\rh{A_{ii}+A_{jj}-2A_{ij}} +\frac1{2m} \sum_{k=1}^m \rh{A_{ii}+A_{kk}-2A_{ik}}\\
		& \quad+ \frac1{2m} \sum_{k=1}^m \rh{A_{jj}+A_{kk}-2A_{jk}} - \frac1{2m^2}\sum_{k=1}^m\sum_{\ell=1}^m \rh{A_{kk}+A_{\ell \ell}-2A_{k\ell}}\\
		& =A_{ij} +\frac1{2m} \sum_{k=1}^m \rh{A_{kk}-2A_{ik}}\\
		& \quad+ \frac1{2m} \sum_{k=1}^m \rh{A_{kk}-2A_{jk}} - \frac1{2m^2}\sum_{k=1}^m\sum_{\ell=1}^m \rh{A_{kk}+A_{\ell \ell}-2A_{k\ell}}\\
		& =A_{ij} - \frac1{m} \sum_{k=1}^m A_{ik} - \frac1{m} \sum_{k=1}^m A_{jk} + \frac1{m^2}\sum_{k=1}^m\sum_{\ell=1}^m A_{k\ell}\\
		 & =\sW(A)_{ij}. 
	\end{align*}
	It follows that $-\frac12\sW\circ \sD=\sW$.
	
	Let $D=\sD(A)$. Note that $D_{ii}=A_{ii}+A_{ii}-2A_{ii}=0$. It follows that $\sD(D)_{ij}=D_{ii}+D_{jj}-2D_{ij}=-2D_{ij}$, so $\sD\circ \sD=-2\sD$. 
	
	Note that 
	\begin{align*}
		\sW\circ \sW = & (-\frac12\sW\circ \sD) \circ \sW = -\frac12\sW\circ( \sD\circ \sW)
		=  -\frac12\sW\circ \sD
		=  \sW.
	\end{align*}

Let $B\in\im(\sD)$. Then, there is an $A\in\mathbb{S}_m$ such that $B=\sD(A)$. Hence, \begin{align*}
	 -\frac12\sD(\sW(B)) =  (\sD\circ -\frac12\sW\circ \sD)(A)= (\sD\circ \sW) (A)=\sD(A)=B. 
\end{align*}
\emph{Vice versa}, let $B\in \im(\sW)$. Then, there is an $A\in\mathbb{S}_m$ such that $B=\sW(A)$. Hence,
  \begin{align*}
	\sW(-\frac12\sD(B))= (-\frac12\sW\circ \sD \circ \sW)(A)=(\sW\circ \sW)(A)=\sW(A)=B.
\end{align*}
It follows that 
\[
\sW:\im(\sD)\to \im(\sW)
\]
is invertible with inverse \[
-\frac12\sD:\im(\sW)\to \im(\sD). 
\]

It follows from $\sD=\sD\circ \sW$ that $\ker(\sW)\subseteq \ker(\sD)$ and it follows from $\sW=-\frac12\sW\circ \sD$, that $\ker(\sD)\subseteq \ker(\sW)$. Hence $\ker(\sD)=\ker(\sW)$.

To calculate the kernel of $\sW$ and $\sD$ we make use of $\sW$.  
Using  \cref{eq:WA}, $\sW(E)=(I-E/m)E(I-E/m)=(I-E/m)(E - E)=0$, so $E\in\ker(\sW)$. Note that $E=e(e/2)^t+(e/2)e^t$. 
	
	Let $v\in\R^m$ satisfy $\sum_{i=1}^m v_i=0$. Then, 
	\begin{align*}
		(I-E/m)(ev^t+ve^t)(I-E/m)
		  &=(ev^t+ve^t-ev^t-0)(I-E/m)\\
		&=  ve^t(I-E/m) 		=  ve^t - ve^t =0. 
	\end{align*}
	It follows that the kernel of  $\sW$ contains $\set{ev^t+ve^t\colon v\in \R^m}$. 
	
Now suppose $A$ is an arbitrary matrix in the kernel of $\sW$. Then we might write $A=\lambda E+ F$, where $\lambda\in\R$ and $F$ is orthogonal to $E$ in the Frobenius inner product, from which follows that $\sum_{i=1}^m\sum_{j=1}^m F_{ij}=0$,  equivalent to $EFE=0$. Moreover, $0=\sW(\lambda E+ F)=\sW(F)$. It follows that 
\begin{align*}
0&=(I-E/m)F(I-E/m)\\
&=(F-EF/m)(I-E/m)\\
&=F-EF/m-FE/m +EFE/(m^2)\\
&= F-EF/m-FE/m.
\end{align*}
That is, $F=EF/m+FE/m$.  Note that  
\begin{align*}
	(EF)_{ij}= & \sum_{k=1}^m F_{kj}
	\end{align*}
does not depend on $i$. 
So there is a vector $x\in\R^m$, so that
 \[
EF =\begin{pmatrix}
	x^t \\ \vdots \\ x^t 
\end{pmatrix}. 
\]  
And we have $FE=(EF)^t=(x\ldots x)$.
 \[
F=\frac1m \begin{pmatrix}
	x^t \\ \vdots \\ x^t
\end{pmatrix} + \frac1m(x\ldots x)=ex^t/m+xe^t/m. 
\]
 It follows that $\ker(\sW)=\set{ev^t+ve^t\colon v\in \R^m}$. The kernel has dimension $m$. Since $\dim(\mathbb{S}_m)=m(m+1)/2$, it follows by the rank-nullity theorem that $\dim(\im(\sW))=m(m-1)/2$. 

It follows from $(-\frac12\sD)\circ(-\frac12 \sD)=-\frac12\sD$ and $\sW\circ \sW=\sW$ that $-\frac12\sD$ and $\sW$ are projections. 

Next we demonstrate that $\sW$ is an orthogonal projection by showing that the image space of $\sW$ is orthogonal to the kernel of $\sW$. Let $B$ be a symmetric $m\x m$-matrix. Then $B$ is orthogonal to the kernel if and only if for all $v\in\R^m$, 
\begin{equation*}
\begin{split}
0& =\inpr{B,ev^t+ve^t}= \sum_{i=1}^m \sum_{j=1}^m B_{ij}(v_i+v_j)\\
&=  \sum_{i=1}^m v_i\sum_{j=1}^m B_{ij} + \sum_{j=1}^m v_j\sum_{i=1}^m B_{ij}= 2 \sum_{i=1}^m v_i \sum_{j=1}^m B_{ij}.  
\end{split}
\end{equation*}
Note that $\set{e_ie^t+ee_i^t:1\le i \le m}$ is a basis for the kernel (where $e_i$ is the $i$th unit  vector), and 
\[
\inpr{B,e_ie^t+ee_i^t} =2\sum_{j=1}^mB_{ij},\quad i=1,\ldots,m.
\]
Thus $B$ is orthogonal to $\ker(\sW)$ if and only if all rows of $B$ sum to zero.

Denote $W=\sW(A)$. Note that 
\begin{align*}
    \sum_{j=1}^m W_{ij} = & \sum_{j=1}^m A_{ij} - \sum_{k=1}^m A_{ik} - \frac1m\sum_{j=1}^m \sum_{k=1}^mA_{jk} + \frac1m\sum_{k=1}^m\sum_{\ell=1}^m A_{k\ell}=0,
\end{align*}
for all $i\in\set{1,\ldots,m}$. 
It follows that $\sW$ is an orthogonal projection. Consequently, the operator norm is one.

From the fact that $-\frac12\sD$ and $\sW$ have the same kernel, and  $\im(\sW)\not=\im(\sD)$ (elements of  $\im(\sD)$ has zero diagonal), it follows from unicity of orthogonal projections that $-\frac12\sD$ cannot be an orthogonal projection. 

Finally, let us calculate the operator norm of $-\frac12\sD$. We prove $\|-\frac12\sD\|_{\text{op}}=\sqrt m$, by showing that $\sqrt m$ is both a lower and an upper bound for $\|-\frac12\sD\|_{\text{op}}$.

First we prove that $\sqrt m$ is a lower bound of the operator norm. Note that   $-\frac12\sD(I- E/m)_{ij} = -1$ when $i\neq j$ and zero otherwise. So $\|-\frac12\sD(I-E/m)\|_F^2  = m(m-1)$. Note that $\|I - E/m\|_F=\sqrt{m-1}$. So $ \|-\frac12\sD\|_{\text{op}} \ge \frac{\|-\frac12\sD(I-E/m)\|_F}{\|I-E/m\|_F}=\frac{\sqrt{m(m-1)}}{\sqrt{ m-1}}= \sqrt{m}$. 

Let us continue with the upper bound. 
Let $A$ be an $m\x m$-matrix of Frobenius norm one. 
We can write $A=B+D$, where $B_{ij}=A_{ij}$ when $i\neq j$ and $B_{ii}=0$, and $D$ is a diagonal matrix with $D_{ii}=A_{ii}$, for all $i,j=1,\ldots,m$. Note that $B$ and $D$ are orthogonal with respect to the Frobenius inner product. Then by the linearity of $\sD$ and  the triangle inequality 
\begin{align*}
   \| -\frac12\sD(A)\|_F \le  \|-\frac12\sD(B)\|_F+\|\frac12\sD(D)\|
    = \|B\|_F+\|-\frac12\sD(D)\|_F.
\end{align*}

We can write \[
D = \sum_{i=1}^m  A_{ii}e_ie_i^T. 
\]
Note that $\set{e_ie_j^T:i,j\in\set{1,\ldots,m}}$ is an orthonormal basis of the space of $m\x m$-matrices in the Frobenius inner product. We have 
\begin{align*}
    -\frac12\sD(e_ke_k^T) = e_ke_k^T - \frac12ee_k^T - \frac12 e_ke^T.
\end{align*} So\begin{align*}
\inpr{-\frac12\sD(e_ae_a^T),-\frac12\sD(e_be_b^T)}= &  \begin{cases}
    \frac12 & \text{ if } a\neq b,\\
    \frac12(m-1) & \text{ if } a=b. 
\end{cases} \\
\end{align*}
We have 
\begin{align*}
    - \frac{1}{2}\sD(D) =  -\frac12 \sD\rh{\sum_{i=1}^m A_{ii} e_ie_i^T }
    =  -\frac12 \sum_{i=1}^m A_{ii} \sD\rh{ e_ie_i^T }
\end{align*}
So 
\begin{align*}
    &\norm{-\frac12\sD(D)}_F^2 = \inpr{-\frac12 \sum_{i=1}^m A_{ii} \sD\rh{ e_ie_i^T },-\frac12 \sum_{i=1}^m A_{ii} \sD\rh{ e_ie_i^T }} \\
    = & \frac12\sum_{i,j:i\neq j}A_{ii}A_{jj} + \frac12(m-1)\sum_{i=1}^m A_{ii}^2. 
\end{align*}
Using that $A_{ii}A_{jj}\le \frac12A_{ii}^2+\frac12A_{jj}^2$, we see that 
\begin{align*}
    \norm{-\frac12\sD(D)}_F^2  \le & \frac14\sum_{i,j:i\neq j}(A_{ii}^2+A_{jj}^2) + \frac12(m-1)\sum_{i=1}^m A_{ii}^2 \\
    = &  \frac14\sum_{i=1}^m\sum_{j=1}^m(A_{ii}^2+A_{jj}^2) + \frac12(m-2)\sum_{i=1}^m A_{ii}^2 \\
    = & (m-1)\sum_{i=1}^m A_{ii}^2\\
    = & (m-1)\|D\|_F^2. 
\end{align*}
As $\|A\|_F=1$, and $B$ and $D$ are orthogonal, we have $1=\|B\|_F^2 + \|D\|_F^2$. Let $\alpha= \|D\|_F^2\in[0,1]$, then $ \|B\|_F^2 = 1-\alpha$. 
So 
\begin{align*}
     \| -\frac12\sD(A)\|_F \le  \sqrt{1-\alpha} + \sqrt{(m-1)\alpha}=:f(\alpha). 
\end{align*}
With simple algebra one can show that the maximum of $f$ is attained for $\alpha=\frac{m-1}{m}$.  
So \begin{align*}
    \|-\frac12 \sD(A)\|_F \le &  \sqrt{\frac1m} + \sqrt{\frac{(m-1)^2}{m}}
  =  \sqrt m.  
\end{align*}
It follows that both the upper and lower bound of  $\|-\frac12\sD\|_{\text{op}}$ are  $ \sqrt m$, so $\|-\frac12\sD\|_{\text{op}} = \sqrt m$. 

\subsection{Proof of  \cref{thm:convergenceofhatV}}\label{subsec:proofthmconvergenceofhatV}

Define  $\widehat V^k = X^k(X^k)^t - \hat\mu_k^2E$ with entries $\widehat V^k_{ab}$, $a,b=1,\ldots,m$, then $\widehat V = \frac1n\sum_{k=1}^n \widehat V^k$. 

Trivially for $x_1,x_2,x_3,x_4\in\mathbb{R}$,
 $|x_1x_2x_3x_4|\le  x_1^4+x_2^4+x_3^4+x_4^4$.
 So we  have for
\begin{align*}
	\widehat V^k_{ab}  &= X_a^kX_b^k - \frac1{m^2} \sum_{c=1}^m\sum_{d=1}^m X_{c}^kX_{d}^k  
\intertext{that}
\E((\widehat V_{ab}^k)^2) &\le E((X_{a}^k)^2(X_{b}^k)^2)+\frac1{m^4} \left(\sum_{c=1}^m\sum_{d=1}^m X_{c}^kX_{d}^k\right)^{\!\!2}+ \frac2{m^2}\E\left( X_{a}^kX_{b}^k\sum_{c=1}^m\sum_{d=1}^m X_{c}^kX_{d}^k\right) \\
&\le 4C+4C+8C=16C,
\end{align*}
where it is assumed that  all  moments of $X_a^k$ up to order four are bounded uniformly in $k=1,\ldots,n$ and $a=1,\ldots,m$, by some number $C>0$. Hence,
\begin{align*}
\var(\widehat V^k_{ab} )&\le \E((\widehat V_{ab}^k)^2)\le 16C.
\end{align*}

As $\widehat V^k,\widehat V^{k+t},\ldots, V^{k+ \floor{(n-k)/t}t}$ are independent, for $k=1,\ldots,t$, we have  
\begin{align*}
\var\rh{\frac1n\rh{\widehat V^k_{ab}+\widehat V_{(k+t)ab}+\ldots+ V_{(k+ \floor{(n-k)/t}t)ab}}}& \le  \frac{\floor{(n-k)/t}}{n^2} 16C\le   \frac{16}{nt}C.
\end{align*}
Applying \cref{lem:boundssecondmomentgeneralm} gives 
\begin{align*}
    \var(\widehat V_{ab}) & \le \frac{16t^2}{nt}C=\frac{16t}{n}C.
\end{align*}
 It follows that 
\begin{align*}
\E[\|\widehat V - V\|_F^2] &=   \sum_{a=1}^m\sum_{b=1}^m \var(\widehat V_{ab}) \le  \frac{16m^2t}n C.
\end{align*}
Consequently,  by Jensen's inequality the claim follows: $\E(\|\widehat V-V\|_F)^2\le \E(\|\widehat V-V\|_F^2)$.  

It follows from \cref{thm:Vestimator} in combination with the definition of $\sD$, that $\widehat D_{ij} = \widehat V_{ii} + \widehat V_{jj} - 2\widehat V_{ij}$.   Similarly, using \cref{lem:boundssecondmomentgeneralm} again and the definition of $\sD$, gives $\var(\widehat D_{ij}) \le \frac {256t}n C$.  Using $\E(\widehat D)=D$, the claim for   $\E(\|\widehat D-D\|_F)$ follows similarly to that for $\widehat V$.

Finally, from \cref{thm:Vestimator},
\[\widehat W_{ij} = \widehat V_{ij} - \frac1m\sum_{a=1}^m \widehat V_{ia} - \frac1m\sum_{a=1}^m \widehat V_{ja} + \frac1{m^2} \sum_{a=1}^m\sum_{b=1}^m \widehat V_{ab}. \]
Applying \cref{lem:boundssecondmomentgeneralm}   gives 
\begin{align*}
    \var(\widehat W_{ij}) & \le 4 \rh{\var(\widehat V_{ij}) + \var\rh{\frac1m\sum_{a=1}^m \widehat V_{ia}}+ \var\rh{\frac1m\sum_{a=1}^m \widehat V_{ja}} + \var\rh{\frac1{m^2} \sum_{a=1}^m\sum_{b=1}^m \widehat V_{ab}}} \\
   &  \le \frac {256t}n C.
\end{align*}
Again, in a similar way to that of $\widehat V$, 
the claim for $\E(\|\widehat W-W\|_F)$ follows.

\subsection{Proof of  \cref{thm:convergenceinL2ofwidetildeSigma}}\label{proof:convergenceinL2ofwidetildeSigma}

Define  $\hat S^k = (X^{2k}-X^{2k-1})(X^{2k}-X^{2k-1})^t$. 
Then we have 
\begin{align*}
    \var(\hat S _{ab}^K) \le & \E\vh{\rh{X^{2k}_a-X^{2k-1}_a}^2\rh{X^{2k}_b-X^{2k-1}_b}^2}\\
    = & \E\vh{ \rh{X_a^{2k}X_b^{2k}-X_a^{2k}X_b^{2k-1}- X_a^{2k-1}X_b^{2k} + X_a^{2k-1}X_b^{2k-1} }^2  }.
\end{align*}
The latter can be written as a sum of 16 elements $\E \vh{X_a^{c}X_a^{d}X_b^{e}X_b^{f}}$, where $c,d,e,f\in\set{2k-1,2k}$. Applying the Jensen's inequality and then  H\"older's inequality twice gives  
\begin{align*}
    \abs{ \E \vh{X_a^{c}X_a^{d}X_b^{e}X_b^{f}}}
    &  \le \E \vh{\abs{X_a^{c}X_a^{d}X_b^{e}X_b^{f}}}\\ 
     & \le  \sqrt{\E \vh{(X_a^{c})^2(X_a^{d})^2}\E\vh{(X_b^{e})^2(X_b^{f})^2}}\\
     & \le   \sqrt[4]{\E \vh{(X_a^{c})^4}\E\vh{(X_a^{d})^4}\E\vh{(X_b^{e})^4}\E\vh{(X_b^{f})^4}}\\
     &\le  \sqrt[4]{C^4}=C. 
\end{align*}
It follows that $\var(\hat S_{ab}^{K})\le 16C$. 

So $\hat S  = \frac1{\floor{n/2}}\sum_{k=1}^{\floor{n/2}} \hat S^k$. 
As $\hat  S^k,$ $\hat S_{k+t},$ $\ldots, \hat S_{k+ \floor{(\floor{n/2}-k)/t}t}$ are independent, for $k=1,\ldots,t$, we have  
\begin{align*}
&\var\rh{\frac1{\floor{n/2}}\rh{\hat S_{ab}^k+\hat S_{ab}^{k+t}+\ldots+ \hat S_{ab}^{k+ \floor{(\floor{n/2}-k)/t}t}}}\\
& \le   \frac{\floor{(\floor{n/2}-k)/t}}{\floor{n/2}^2} 16C\le  \frac{16C}{\floor{n/2}t }.
\end{align*}
Applying \cref{lem:boundssecondmomentgeneralm} gives
\begin{align*}
    \var(\hat S_{ab})  &= \var\rh{\sum_{k=1}^t \frac1{\floor{n/2}}\rh{\hat S_{ab}^k+\hat S _{ab}^{k+t}+\ldots+ \hat S_{ab}^{k+ \floor{(\floor{n/2}-k)/t}t}}} \\
     &\le  t \sum_{k=1}^t \var\rh{\frac1{\floor{n/2}}\rh{\hat S_{ab}^k+\hat S_{ab}^{k+t}+\ldots+ \hat S_{ab}^{(k+ \floor{\floor{n/2}-k)/t}t}}}\\
    &\le    t \sum_{k=1}^t \frac{16C}{\floor{n/2}t }\\
     & =\frac{16C}{\floor{n/2}}t.
\end{align*}
 It follows that 
\begin{align*}
\E[\|\hat S -  \Sigma_1 \|_2^2] = &  \sum_{a=1}^m\sum_{b=1}^m \var(\widetilde \Sigma_{ab}) \le  \frac{16m^2t}{\floor{n/2}} C,
\end{align*}
and by Jensen's inequality that 
\[\E[\|\hat S- \Sigma_1 \|_2]\le \sqrt{\E[\|\hat S -  \Sigma_1 \|_2^2]}\le 4\sqrt{\frac{m^2t}{\floor{n/2}} C}.\]

\section{Auxiliary results}

\begin{lemma}\label{lem:boundssecondmoment}
Let $x,y\in\R$. Then $ |xy|\le\frac 12 (x^2+y^2)$.
\end{lemma}
\begin{proof}It follows from $0\le(x-y)^2=x^2+y^2-2xy$ and $0\le(x+y)^2=x^2+y^2+2xy$ that $|2xy|\le x^2+y^2$.
\end{proof}

\begin{lemma}\label{lem:boundssecondmomentgeneralm1}
Let $x_1,\ldots,x_m\in\R$. Then $(\sum_{i=1}^mx_i)^2\le m\sum_{i=1}^mx_i^2$.
\end{lemma}
\begin{proof}From  \cref{lem:boundssecondmoment}, 
$(\sum_{i=1}^mx_i)^2=\sum_{i=1}^m\sum_{j=1}^m x_ix_j\le \frac12\sum_{i=1}^m\sum_{j=1}^m(x_i^2+x_j^2)=m\sum_{i=1}^mx_i^2$.
\end{proof}

\begin{corollary}\label{lem:boundssecondmomentgeneralm}
  Let $X_1,\ldots,X_m$ be random variables. Then,  $\textrm{Var}(\sum_{k=1}^m X_k)$ $ \le  m \sum_{k=1}^m \textrm{Var}(X_k)$.
\end{corollary}
\begin{proof}
    Take $x_k=X_k-\E[X_k]$ and expectation in \cref{lem:boundssecondmomentgeneralm1}. 
\end{proof}

\section*{Acknowledgements}
CW and JvW are supported by the Independent Research Fund Denmark (grant number:  8021-00360B) and the University of Copenhagen through the Data+ initiative.  ZI is supported by the Novo Nordisk Foundation, Denmark (grant number: NNF20OC0061343).

\bibliographystyle{DeGruyter}
\bibliography{Bieb}

\end{document}